\documentclass[onecolumn,draftclsnofoot,12pt]{IEEEtran}
\usepackage{cite} 
\usepackage{epsfig}
\usepackage{epstopdf}
\usepackage{graphicx}
\usepackage{url}
\usepackage{bm}
\usepackage{amsfonts}
\usepackage{amsmath,bm,amsthm}
\usepackage{amssymb}
\usepackage{multicol}
\usepackage{times}
\usepackage{psfrag}
\usepackage{subfigure}
\usepackage{stfloats}
\usepackage{geometry}
\usepackage{array}
\usepackage{booktabs, threeparttable}
\usepackage{setspace}
\usepackage{color}
\usepackage{inputenc}
\usepackage{gensymb}

\newtheorem{theorem}{Theorem}

\allowdisplaybreaks

\newtheorem{corollary}{Corollary}
\usepackage{algorithm}
\usepackage{algorithmic}

\newcommand{\diag}{{\rm diag}}
\newcommand{\p}{{\Omega}}

\newcommand{\hc}{{{\bf h}^{\rm C}_u}}
\newcommand{\hs}{{{\bf h}^{\rm S}}}

\newcommand{\vc}{ {\rm vec}}
\newcommand{\rep}{ {\rm Re}}
\newcommand{\imp}{ {\rm Im}}

\newcommand{\rx}{ {\bf R}}
\newcommand{\pu}{ {\bf p}}

\newcommand{\pr}{ {\bf P}}

\begin{document}
\title{Waveform Optimization with Multiple Performance Metrics for Broadband Joint Communication and Radar Sensing}

\author{
{Zhitong Ni,~\IEEEmembership{Student~Member,~IEEE},
 J. Andrew Zhang,~\IEEEmembership{Senior~Member,~IEEE},\\
   Kai Yang,~\IEEEmembership{ Member,~IEEE},
    Xiaojing Huang,~\IEEEmembership{Senior~Member,~IEEE},\\
     and Theodoros A. Tsiftsis},~\IEEEmembership{Senior~Member,~IEEE}\\
\thanks{  Z. Ni, K. Yang are with the School of Information and Electronics, Beijing Institute of Technology, Beijing, 100081, China. Z. Ni is also with the Global Big Data Technologies Centre, University of Technology Sydney, NSW, 2007, Australia  (Emails: zhitong.ni@student.uts.edu.au, yangkai@ieee.org).}
\thanks{ J. Andrew Zhang and X. Huang are with the Global Big Data Technologies Centre, University of Technology Sydney, NSW, 2007, Australia (Emails:  Andrew.Zhang@uts.edu.au, xiaojing.huang@uts.edu.au).}
\thanks{T. A. Tsiftsis is with the Institute of Physical Internet, and with the school of Intelligent Systems Science and Engineering, Jinan University, Zhuhai Campus, Zhuhai, 519070, China. (Email: theo\_tsiftsis@jnu.edu.cn).}
}	
	
\maketitle

\begin{abstract}
Joint communication and radar sensing (JCAS) integrates communication and radar/radio sensing into one system, sharing one transmitted signal. In this paper, we investigate JCAS waveform optimization underlying communication signals, where a base station detects radar targets and communicates with mobile users simultaneously. We first develop individual novel waveform optimization problems for communications and sensing, respectively. For communications, we propose a novel lower bound of sum rate by integrating multi-user interference and effective channel gain into one metric that simplifies the optimization of the sum rate. For radar sensing, we consider optimizing one of two metrics, the mutual information or the Cram\'{e}r-Rao bound. Then, we formulate the JCAS problem by optimizing the communication metric under different constraints of the radar metric, and we obtain both closed-form solutions and iterative solutions to the non-convex JCAS optimization problem. Numerical results are provided and verify the proposed optimization solutions.
\end{abstract}


\section{Introduction}
Sharing many hardware and signal processing modules and transmitting one single signal, joint communication and radar sensing (JCAS) systems have shown great potential in many applications. One main application is for the fifth-generation (5G) cellular networks or beyond \cite{abdelhadi16,lushan5G}, where the transmitted signals are used for both providing communication services and sensing the surrounding environments. JCAS also brings a lot of interests in vehicular networks and self-driving \cite{kumari,daniel18,dokhanchi}, where signals used for communications between cars are also used for sensing the environment to detect objects and avoid collisions.

\subsection{Related Works}
With most traditional hardware components replaced by digital processing modules, JCAS-enabled systems can reduce the number of connected devices and save the frequency resources \cite{strum1,chiri17,liu1}. To utilize the resources completely, most developed JCAS systems transmit single waveforms \cite{liuyongjun}, except some works realizing JCAS functions via multiplexing technologies, such as time division multiplexing (TDM)  \cite{LHan11_gud2}, frequency division multiplexing (FDM) \cite{liuyongjun9}, or code division multiplexing (CDM) \cite{liuyongjun10}.
The resource utilization of single waveforms is highly efficient since the waveforms are shared by both communication and radar systems.
Such waveforms can be developed from classical radar waveforms \cite{liuyongjun12} or conventional communication waveforms \cite{sit2014mimo}.

Single waveform optimization has been reported in the literature.
In \cite{andrew19}, a multi-beam approach was proposed to flexibly generate JCAS sub-beams using analog antenna arrays. The optimization of the multi-beams was further investigated in \cite{Luo19}. This approach can adapt to varying channels but is suboptimal. In \cite{liumu}, the authors separated antenna arrays into two groups for realizing dual-function JCAS systems. The radar waveform falls into the nullspace of the downlink channel, such that the interference between radar signals and communication signals can be minimized. A multi-objective function was further applied to trade off the similarity between the generated waveform and the desired one in \cite{liuweight}. The multi-objective function in  \cite{liuweight} is a weighted sum of two individual optimal waveforms.
In \cite{libo}, the authors maximized the radar signal-to-interference-plus-noise ratio (SINR) with a given specific capacity of communication channels. The work in \cite{libo2,libo3} further introduced a sub-sampling matrix for radar as an objective function of the optimization. It is noted that all these works studied waveform optimization in a narrowband scenario and ignored the frequency diversity.

When optimizing JCAS waveforms,  different performance metrics for both communications and radar sensing can be used. The metrics for multiuser communication systems include multiuser interference (MUI) \cite{liuweight,sit} and  effective channel gain (ECG) \cite{niTVT20}. The MUI tolerance was analyzed for multi-input-multi-output orthogonal-frequency-division-multiplexing (MIMO-OFDM) JCAS systems in \cite{sit}, using the interleaved signal model in \cite{strum13}. The individual optimal communication waveform in \cite{liuweight} is also based on minimizing the MUI. However, simply minimizing MUI can cause the JCAS signals to fall into the nullspace of individual communication signals and result in a low SINR. To tackle this issue, the authors in \cite{liuweight} introduced an expected diagonal matrix that requires prior knowledge about the ECG.

As for radar sensing,  typically considered performance metrics include mutual information (MI) \cite{Liu17letter,inner12}, Cram\'{e}r-Rao bound (CRB) \cite{liuyongjun,Innerbound}, and minimum mean-squared error (MMSE) \cite{coexist11}. In \cite{Liu17letter}, MI for an OFDM JCAS system was studied, and the power allocation on subcarriers was investigated by maximizing the weighted sum of the MI of radar and the MI of communication. In \cite{inner12}, the authors developed an MI measure that jointly optimizes the performance of radar and communication systems that overlap in the same frequency band. The CRB and MMSE are also commonly used metrics for waveform optimization. In \cite{Innerbound}, the authors derived the performance bounds in a single antenna system. In \cite{coexist11}, the authors derived performance bounds for the radar estimation rate based on MMSE estimation bounds.

\subsection{Motivations and Contributions}
Although the above-mentioned metrics have been individually studied for JCAS systems, no links between them have been identified so far. Some of those papers can be invalid when the number of sensing targets is not equal to that of the communication users. For communications, the SINR is the main concern, while most papers do not address the SINR directly and use some indirect metrics. For radar sensing, we note that different metrics can be used but there are no parallel comparisons. Based on the limitations of prior works, we aim to build a common method that can utilize different metrics for both communications and radar sensing.

This paper aims to develop waveform optimization algorithms for broadband JCAS systems and establish connections between different metrics of optimization. We consider a practical JCAS system of deployment. JCAS systems, particularly those underlying communication signals, confront the major challenge of full-duplex, that is, transmitting and receiving at the same time using the same frequency channel \cite{fullduplex}.
Solutions bypassing full-duplex requirements were discussed in \cite{lushan,niICC}. Here, we adopt a single receive antenna dedicated to radar sensing, which is a low-cost and practical solution at the moment.

The key idea is to optimize a developed novel lower bound of SINR for communications under different constraints of the radar metrics. By studying optimization with respect to these metrics together, for the first time, we disclose the important connections and relationships between these metrics in JCAS systems.

The main contributions of this paper are summarized as follows.
\begin{itemize}
\item We derive multiple metrics, including MUI and ECG for communications, and MMSE, MI, and CRB for radar sensing, in a broadband JCAS system. In the adopted system, bypassing the full-duplex requirement, the base station (BS) uses one single receive antenna for collecting the signals reflected from targets.
\item We develop a novel lower bound of SINR as the communication metric, which contains both MUI and ECG. The developed lower bound simplifies the expression of SINR and makes it flexible for joint waveform optimization, involving both communication and radar metrics.
\item For radar sensing, we derive multiple metrics, including MMSE, MI, and CRB. We show the connections between optimizing these metrics.
\item We apply these multiple performance metrics and formulate two kinds of JCAS waveform optimization problems by minimizing the proposed communication metric under the constraint of the radar metric. We obtain the corresponding closed-form solutions under certain conditions. When the conditions are not satisfied, we propose iterative solutions based on Newton's method.
\end{itemize}

Notations: $\rm\bf a$ denotes a vector, $\rm\bf A$ denotes a matrix, italic English letters like $N$ and lower-case Greek letters $\alpha$ are a scalar. $|{\rm\bf A }|, {\rm\bf A }^T, {\rm\bf A }^H, {\rm\bf A }^*,$ and ${\rm\bf A }^\dag$ represent determinant value, transpose, conjugate transpose, conjugate, and pseudo inverse, respectively. We use ${\rm diag}({\bf a})$ to denote a diagonal matrix with diagonal entries being the entries of ${\bf a}$.  $\|{\rm\bf A }\|_F$ and  $[{\rm\bf A }]_{ N}$ represent the Frobenius norm and the $N$th column of a matrix, respectively.

\section{System and Channel Models}

\begin{figure}[t]
    \centering
    \includegraphics[scale=0.8]{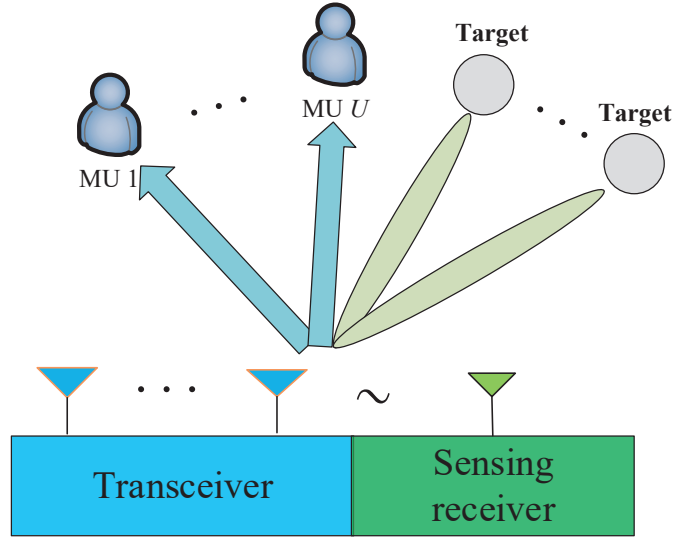}
    \caption{Illustration of JCAS systems with downlink sensing. The BS uses a single-antenna receiver dedicated for sensing the reflected downlink signals.}
    \label{model}
\end{figure}

We consider an OFDM-based JCAS system supporting multiuser communications, as well as active downlink sensing. As shown in Fig. \ref{model}, the sensing is conducted at the BS, where an $N\times 1$ uniform linear array (ULA) is employed for simultaneous communications and radar sensing. Since the full-duplex technologies are not mature yet, we consider a low-cost alternative setup, where a single antenna, sufficiently separated from and synchronized with the transmit array, is used for solely collecting the reflected signals and detecting targets. Note that the sensing parameters can still be effectively estimated by using the received signal from the single-antenna receiver, except for the angle-of-arrivals (AoAs). The angle-of-departures (AoDs) can be obtained and serve a similar purpose as the AoAs.

The OFDM signal consists of $K$ subcarriers with the subcarrier interval being $\Delta f=1/T$, where $T$ is the length of one OFDM symbol exclusive of the cyclic prefix (CP) of length $T_{\rm C}$. Letting ${\bf s}[k,m]$ denote a $U\times 1$ data symbol vector on subcarrier $k$ of the $m$th time slot, we assume the data symbols are statistically independent, i.e., ${\mathbb E}({\bf s}[k,m]{\bf s}^H[k,m])=P {\bf I}$ with $P$ being the transmit power. A digital precoder with the dimension of $N\times U$, $\pr[k,m]$, is applied to ${\bf s}[k,m]$, such that  $\|\pr[k,m]\|_F^2=U$. The precoded data symbol vector is written as
\begin{align}
&{\bf x}[k,m]= \pr[k,m]{\bf s}[k,m], \notag\\
&k\in\{0,\cdots,K-1\}, m\in\{0,\cdots,M-1\}.
\end{align}
An inverse fast Fourier transform (IFFT) across the frequency band is then applied to the elements of ${\bf x}[k,m]$ from $k=0$ to $k=K-1$ and transforms signals into time domain. The time-domain signals are transmitted by the $N\times 1$ ULA.

The BS conducts downlink communications with $U$ multiple users (MUs) that are also potential targets. Each MU has a single receive antenna. It is noted that the system setup is equivalent to a multi-input-single-output (MISO) system. From \cite{niICC} that adopts the same system setup, assembling $m$ from $0$ to $M$ would obtain highly accurate estimates of parameters. Hence, we assume that Doppler frequency approximately equals zero from $m=0$ to $m=M-1$ and make the system and channel model invariant to $m$. Let us send a data block ${\bf S}[k]=\big[{\bf s}[k,0],\cdots, {\bf s}[k,M-1]\big]$ of length $M$ at the baseband of BS. The transmitted signals go through a beam-space multi-path channel. The received signal for each MU is sampled at the interval of $T/K$ and converted to the frequency domain using $K$-point FFT's after removing CP. The frequency-domain received signal block at the $u$th MU is
\begin{align}\label{eq5}
{\bf r}_u[k] = (\hc[k])^H \pr[k] {\bf S}[k] + {\bf n}_u^1[k],
\end{align}
where $ {\bf S}[k]$ is the data block, ${\bf n}_u^1[k]$ is a complex AWGN vector with zero mean and covariance matrix of $\sigma_1^2{\bf I}_M$, and $\hc[k]$ is a communication geometric channel, given by
\begin{align}\label{commch}
\hc [k] =  \sum\limits_{l_u=1}^{L_u}\beta_{l_u}{\bf a}(\Phi_{l_u}) e^{-j2\pi k\frac {\tau_{l_u}} T},
\end{align}
where $\beta_{l_u}$ is the channel gain of the $l_u$th path between BS and MU $u$, $\Phi_{l_u}=\pi\sin(\varphi_{l_u})$ is the equivalent AoD with $\varphi_{l_u}$ being the actual AoD, ${\bf a}(\cdot)$ is an $N\times 1$ normalized ULA array response vector, and $\tau_{l_u}$ is the time delay.

Meanwhile, the BS utilizes the single-antenna receiver to perform radar sensing. We assume that the transmitted signals from $N$ antennas impinge on $L$ targets with the time delays of $\{\tau_l\}_{l=1}^L$. Let us denote the sensing channel as
\begin{align}
\hs[k] & = \sum\limits_{l=1}^L \alpha_l{\bf a}(\p_l) e^{-j2\pi k\frac{\tau_l} T}\notag\\
&={\bf A}{\bf X}{\bf W}[k]{\bf 1},
\end{align}
where $\alpha_l$ is the path loss coming from the $l$th target, $\p_l=\pi\sin(\omega_l)$ is the equivalent AoD with $\omega_l$ being the AoD, $\tau_l$ is the delay of the $l$th target, ${\bf A}=[{\bf a}(\p_1),\cdots,{\bf a}(\p_L)]$, ${\bf X}=\diag(\alpha_1,\cdots, \alpha_L  )$, ${\bf W}[k]=\diag([e^{-j2\pi k\frac{\tau_1} T},\cdots, e^{-j2\pi k\frac{\tau_L} T}])$, and  ${\bf 1}$ is an $L\times1$ vector with all entries being $1$s.
We assume that the entries of radar channel are  independent and identically distributed (i.i.d.).
In \cite{niICC},  a parameter sensing algorithm has been proposed to estimate the parameters using the adopted setup. Hence, in this paper, we assume that both radar and communication channels are known at the BS by using the estimation algorithms in \cite{niICC}, and focus on the waveform optimization.

The received frequency-domain signal at the BS for radar sensing is written as
\begin{align}\label{eq4}
 {\bf r}[k]=(\hs[k])^H{\bf P}[k]{\bf S}[k] +{\bf n}[k],
\end{align}
where ${\bf n}[k]$ is a $1\times M$ complex additive-white-Gaussian-noise (AWGN) vector with zero mean and covariance matrix of $\sigma^2{\bf I}_M$.
\section{Individual Performance Metrics}\label{Ind}
In this section, we will derive the performance metrics for communications and sensing individually. For communications, due to the non-convex nature of SINR, we propose and maximize a lower bound of SINR that is up to the MUI and the ECG of MUs. For radar sensing, we consider three different metrics, MMSE, MI, and CRB, and aim to disclose the links between these metrics.

\subsection{Metric for Communications}
The sum rate or the SINR is the main concern for communications. However, the SINR has a non-convex form and can be difficult to optimize, especially when JCAS functions are required. In JCAS waveform optimization problems, both MUI and ECG have a dominant impact on the sum rate of communications. The MUI is on the denominator of SINR and can be expressed as
\begin{align}\label{MUI1}
{\rm MUI} & = \sum\limits_{k=1}^K\sum_{u=1}^U \sum_{v\neq u}^U \left|({\bf h}^{\rm C}_v[k])^H\pu_u[k]\right|^2\notag\\
          & =  \sum\limits_{k=1}^K\left\|{\bf H}^H[k]{\pr[k]}-{\rm diag}({\bf H}^H[k]{\pr[k]})\right\|_F^2\notag\\
          &\triangleq \sum\limits_{k=1}^K \sum\limits_{u=1}^U I_{k,u},
\end{align}
where ${\bf H}[k]=[{\bf h}^{\rm C}_1[k],\cdots,{\bf h}^{\rm C}_U[k]]$, $\pu_u[k]$  is the $u$th column of $\pr[k]$, and $I_{k,u}$ denotes as the MUI of the $u$th MU on the $k$th subcarrier. Traditional optimization methods that mitigate MUI for communications only are not appropriate for JCAS problems, since only minimizing MUI can cause that the precoding vectors fall into the nullspace of ${\bf H}[k]$, and result in a low ECG.

The ECG of MU-MISO systems is on the nominator of the SINR and can be written as
\begin{align}\label{ECG1}
{\rm ECG}  &=\sum\limits_{k=1}^K\left\| {\rm diag}({\bf H}^H[k]{\pr[k]})\right\|_F^2\notag\\
          &\triangleq\sum\limits_{k=1}^K \sum\limits_{u=1}^U S_{k,u},
\end{align}
where $S_{k,u}$ denotes the ECG of the $u$th MU on the $g$th subcarrier.
To make the following JCAS problems flexible and easy to solve, we develop a regulated bound of SINR as
\begin{align}\label{JJ}
J =&{\rm ECG}- \mu {\rm MUI}\notag\\
=&\sum\limits_{k=1}^K \sum\limits_{u=1}^U S_{k,u}-\mu I_{k,u}\notag\\
 =&\sum\limits_{k=1}^K \sum\limits_{u=1}^U  \left(\left\|(\hc[k])^H{\pu_u}[k]\right\|_F^2-\mu \left\|\tilde {\bf H}_u^H[k]{\pu}_u[k]\right\|_F^2\right)\notag\\
 =&-\sum\limits_{k=1}^K \sum\limits_{u=1}^U  \pu_u^H[k]\rx_u[k]\pu_u[k],
\end{align}
where $\tilde {\bf H}_u[k]=\left[{\bf h}^{\rm C}_1[k],\cdots, {\bf h}^{\rm C}_{u-1}[k],{\bf h}^{\rm C}_{u+1}[k], \cdots, {\bf h}^{\rm C}_U[k]\right]$, $\rx_u[k]=\mu\tilde{\bf H}_u[k]\tilde{\bf H}_u^H[k]-\hc[k](\hc[k])^H$, and $\mu$ is a weighting coefficient to balance between ECG and MUI.

The defined regulated bound, $J$, can be seen as a lower bound of SINR when two conditions are satisfied.  The first condition is $S_{k,u}>\mu I_{k,u}$. The second condition is $I_{k,u}+U\sigma_1^2/P\leq 1$.  Generally, $U\sigma_1^2/P$ can be neglected at high SNRs. Hence, both conditions can be satisfied at high SNRs by letting $\pr[k]=({\bf H}[k])^\dag$, that is, we need $I_{k,u}$  to be zero.

\begin{theorem}\label{T1}
 $J$ is a lower bound of SINR when $I_{k,u}+U\sigma_1^2/P\leq1$ and $S_{k,u}>\mu I_{k,u}$, $\forall u,k$.
\end{theorem}
\begin{proof}
The proof is shown in Appendix \ref{App0}.
\end{proof}
Theorem \ref{T1} unfolds that the defined metric, $J$, is a lower bound of SINR as long as the SNR is sufficiently large and the scaled MUI is less than the ECG.
We also note that, when the noise term is zero, the SINR is no smaller than $KU\log_2(1+\mu)$. Hence, $\mu$ should be much larger than one at high SNRs to guarantee a high SINR.

By using $J$ to jointly optimize MUI and ECG, the individual waveform optimization problem for communications is formulated as
\begin{align}\label{Idicomm}
{\rm \arg}\, & \mathop{\min}\limits_{\pr[k]}  - J =\sum\limits_{k=1}^K \sum\limits_{u=1}^U  \pu_u^H[k]\rx_u[k]\pu_u[k]\notag\\
{\rm s.t.}& \|\pr[k]\|_F^2 \leq U\notag\\
&S_{k,u}-\mu I_{k,u}>0\notag\\
&I_{k,u}+U\sigma_1^2/P\leq 1.
\end{align}
The minimization of $-J$ is still a non-convex problem since $\rx_u[k]$ is not semi-definite, but we note that $J$ has a much simpler expression than the SINR. Hence, When multiple metrics including both communication and radar are considered, $J$ can be a good substitute for the SINR.
By performing the eigen-decomposition of $\rx_u[k]$, i.e., $\rx_u[k]={\bf V}_u[k]{\bf E}_u[k]{\bf V}_u[k]^H$,  we note that there exists one negative diagonal entry in ${\bf E}_u[k]$. The optimal solution of $\pu_u[k]$ should be the eigenvector with the corresponding eigenvalue being negative.

\subsection{MI and MMSE for Sensing}\label{MIOP}
MI and MMSE are widely-used metrics in radar systems. Maximizing MI enables the maximization of information in the received signal, which contains the parameters of targets during its propagation. Minimizing MMSE is to obtain the parameters with the highest accuracy. With obtaining the accurate information about the targets, the parameters related to sensing can be extracted using various sensing algorithms.

The MI of the received OFDM symbols for sensing can be represented as
 \begin{align}
{\rm MI}=& \sum\limits_{k=0}^{K-1} \tilde H( {\bf y}[k]; \hs [k]|{\bf S}[k] ) \notag\\
=&\sum\limits_{k=0}^{K-1}\log_2 \left(\left|{\bf I}_U+\frac{ P\pr^H[k]\hs[k](\hs[k])^H\pr[k]}{\sigma^2}\right|\right)\notag\\
=& \sum\limits_{k=0}^{K-1}  \tilde H(  \hs [k];{\bf y}[k]|{\bf S}[k] )\notag\\
=&\sum\limits_{k=0}^{K-1}\log_2 \left( 1+\frac{ P(\hs[k])^H\pr[k]\pr^H[k]\hs[k]}{\sigma^2} \right)\label{convMI}.
 \end{align}
From \eqref{convMI}, it is clear that maximizing MI is a convex problem and the optimal $\pr[k]$ satisfies that $\hs[k]=[{\bf U}_{\rm P}[k]]_1$, where ${\bf U}_{\rm P}[k]$ is the left singular matrix of $\pr[k]$, that is, any $\pr[k]$ satisfying $[{\bf U}_{\rm P}[k]]_1=\hs[k]$ is the optimal solution. The globally optimal $\pr[k]$ should maximize the first singular value of $\pr[k]$. Hence, the optimal MI precoder has a rank of one. To make the precoder of dimension $N\times U$ have only one rank, the optimal MI precoder is obtained as
\begin{align}
\pr^\star_{\rm MI}[k]&= \hs[k]{\bf \Lambda}[k],
\end{align}
where ${\bf \Lambda}[k]$ is an arbitrary $1\times U$ matrix satisfying $\|{\bf \Lambda}[k]\|_F^2=\frac{U}{\| \hs[k]\|_F^2}$.
Note that the  size of the precoder is determined by the number of users instead of number of targets, since the number of targets can be varied from time to time.

The MMSE can be written as
 \begin{align}
{\rm MMSE} = \|\hs[k]-{\hat{\bf h}}^{\rm S}[k]\|^2,
 \end{align}
where ${\hat{\bf h}}^{\rm S}[k]$ is the estimated radar channel.
Referring to the derivations in \cite{guo2005mutual}, when $\hs[k]$ is Gaussian distributed, we have the result of $\frac {\partial {\rm MI}}{\partial P}=\frac12 {\rm MMSE}$. Hence, the larger the transmit power is, the less the MMSE becomes. Due to the existence of noise, the MMSE reaches its minimum value that is larger than zero. Since MMSE is related to different sensing algorithms, we will use MI instead of MMSE as one radar metric for the following JCAS problem formulation.

\subsection{CRB for Sensing}\label{CRBOP}

CRB is a theoretical lower bound of MMSE. Note that CRB is not up to the sensing algorithms and generally cannot be achieved.
In our adopted broadband MISO systems, we mainly consider the accuracy of estimates for delays, complex gains, and AoDs of the targets. The Doppler frequency is assumed to be zero in the training period. Note that the BS system setup is not a typical mono-static radar. We assume that it is an approximate mono-static radar for the convenience of deriving the CRB. We rewrite the received radar signal of \eqref{eq4} into the form of vector, i.e.,
\begin{align}
{\bf y}^T[k]&= (({\hs}[k])^H\pr[k]{\bf S}[k]+{\bf n}[k])^T\notag\\
&=({\bf 1}^H{\bf W}^H[k]{\bf X}{\bf A}^H\pr[k]{\bf S}[k])^T+{\bf n}^T[k]\notag\\
&={\bf S}^T[k]\vc({\bf 1}^H{\bf W}^H[k]{\bf X}{\bf A}^H\pr[k])+{\bf n}^T[k]\notag\\
&\triangleq {\bf S}^T[k] {\bf q}[k] + {\bf n}^T[k],
\end{align}
where   ${\bf q}[k]= \pr^T[k]{\bf A}^*{\bf X}^*{\bf W}^*[k]{\bf 1}$ denotes the equivalent radar channel. Note that ${\bf q}[k]$ contains all parameters of interests, i.e., delays, complex gains, and AoDs, and contains the precoder to be determined. For convenience, we represent all parameters collectively using one common real-valued vector ${\bf \Theta}=\big[\rep [\alpha_1,\cdots, {\alpha}_L],\imp[\alpha_1,\cdots,\alpha_L],\tau_1,\cdots,\tau_L, \p_1,\cdots, \p_L\big]^T$. It is noted that there are $4L$ parameters in total. Hence, ${\bf q}[k]$ is a function with $4L$ variables.

The CRB matrix, denoted as ${\bf B}$, contains all lower bounds of the estimates. The $(i,j)$th entry of ${\bf B}$ denotes the impact of the parameter $i$ on parameter $j$. The total CRB can be expressed as the Frobenius norm of ${\bf B}$. The CRB matrix is the inverse matrix of Fish information matrix (FIM), i.e., ${\bf B}={\bf F}^{-1}$.  Referring to the derivations of \cite{perfoB} and using the  so-called {\textit{ DOD/DOA Model}} defined in \cite{perfoB}, we obtain our corresponding FIM as
\begin{align}
{\bf F}=\left[\begin{array}{cccc}
\rep[{\bf F}_1]&-\imp[{\bf F}_1]&\rep[{\bf F}_2]&\rep[{\bf F}_3]\\
\imp[{\bf F}_1]& \rep[{\bf F}_1]&\imp[{\bf F}_2]&\imp[{\bf F}_3]\\
\rep[{\bf F}_2^T]& \imp[{\bf F}_2^T]&\rep[{\bf F}_4]&\rep[{\bf F}_5]\\
\rep[{\bf F}_3^T]& \imp[{\bf F}_3^T]&\rep[{\bf F}_5^T]&\rep[{\bf F}_6]\\
\end{array}\right],
\end{align}
where ${\bf F}_1,\cdots,{\bf F}_5$, and ${\bf F}_6$ are given by Appendix \ref{App1}. The derivation can be referred to \cite{perfoB}.

From Appendix \ref{App1}, we note that the CRB matrix depends on the covariance matrix, $ {\bf Q}[k]=\pr[k]\pr^H[k]$, only. Therefore, the CRB-based optimization problem is equivalent to optimizing ${\bf Q}[k]$. As we mentioned above, the total CRB equals the Frobenius norm of  ${\bf B}$, which is equivalent to the sum of the singular values of ${\bf B}$.
We aim to minimize the largest eigenvalue of the CRB matrix as in \cite{CRB_LuzhouXu}. Minimizing the largest eigenvalue of the CRB matrix is equivalent to maximizing the smallest eigenvalue of the FIM. The waveform optimization under the eigen-optimization criterion can be directly formulated as a semi-definite-programming (SDP),
\begin{align}\label{FIMop}
{\rm \arg}\, &\mathop{\min} \limits_{{\bf Q}[k]}  -t \notag\\
{\rm s.t.}  &  {\bf F} \succeq t{\bf I}, {\bf Q}[k]\succeq {\bf 0}, {\rm Tr}({\bf Q}[k])\leq U,
\end{align}
where $t$ is an auxiliary variable. However, the problem above is still difficult to address. We attempt to simplify the problem of \eqref{FIMop}. From Appendix \ref{App1}, we note that the norm of ${\bf F}_4$ has an order of $k^2$, which is significantly larger than that of other matrices in ${\bf F}$. Hence, we can approximately assume that ${\bf F}$ equals $\rep[{\bf F}_4]$, that is, only the accuracy of delays is counted.

\textit{Remark 1:} Comparing MI and CRB, we note that they correspond to two types of optimization metrics. In the MI metric, the optimal waveform matrix can be obtained, while in the CRB metric, the optimal covariance waveform matrix can be obtained. For JCAS problem formulation, these two types of metrics also appear in many other metrics, such as channel capacity, desired semi-definite covariance matrix \cite{liuweight}, and waveform similarity \cite{multimetric}. Hence, we will use these two types of optimal radar metrics and the proposed communication lower bound to formulate JCAS problems.

\section{MI Constrained JCAS Waveform Optimization}\label{sec-JCAS}
In this section, we formulate optimization problems that jointly consider radar performance and communication performance.  Note that in the section \ref{Ind}, we have obtained one communication metric that contains both MUI and ECG, and obtained two radar metrics, one for MI and one for CRB. We will combine the communication metric with one of the radar metrics.

\subsection{Closed-Form Solution}
In this subsection, we constrain the MI of radar and simultaneously maximize the  proposed lower bound of SINR, i.e., $J$. Since the optimal MI is obtained via section \ref{MIOP}, we can define a threshold of MI, $\tilde H_0$, that is smaller than the maximal MI. The MI-constrained JCAS optimization problem is formulated as
 \begin{align}\label{MICONs}
{\rm \arg}\, \mathop{\min}\limits_{\pr[k] }& - J  =\sum\limits_{k=1}^K \sum\limits_{u=1}^U \pu_u[k]^H\rx_u[k] \pu_u[k]\notag\\
& =\sum\limits_{k=1}^K \sum\limits_{u=1}^U \mu I_{k,u}-S_{k,u} \notag\\
{\rm s.t.}& \|\pr[k]\|_F^2\leq U,\notag\\
&\tilde H(  \hs [k];{\bf y}[k]|{\bf S}[k] )\geq \tilde H_0, \notag\\
&  \mu I_{k,u}-S_{k,u}< 0,\notag\\
&I_{k,u}+U\sigma_1^2/P\leq 1.
\end{align}
Note that minimizing the objective function above is not a convex problem, since $\rx_u[k]$ in $J$ is not a semi-definite matrix. Hence, the problem above cannot be solved directly using convex-optimization toolboxes.

Since MI maximization is a convex problem, the range of $\pr[k]$ that satisfies $\tilde H(\hs [k];{\bf y}[k]|{\bf S}[k] )\geq H_0$ should be a convex set. Hence, we can directly transform the MI constraint set into another convex set. For the convenience of following proposed algorithm, we adopt the convex set that limits the Euclidean distance between the JCAS precoder and the optimal MI precoder, i.e.,
\begin{align}
\|\pr[k]- \pr_{\rm MI}^\star[k]\|_F^2\leq \rho,
\end{align}
where $\rho$ is a threshold to constrain the Euclidean distance between the JCAS precoder and the optimal MI precoder. The transformed constraint is a complex sphere.

For the third constraint, we note that it has the same expression as the term in the objective function. As long as we can find an initial value of $\pr[k]$ satisfying $\mu I_{k,u}-S_{k,u}<0$, we can use iterative solutions to make the objective function keep dropping.

For the fourth constraint, we note that the norm of ${\bf H}[k]$ can be an arbitrary value. Hence, as long as $U\sigma^2/P<1$, we can make $I_{k,u}\leq 1-U\sigma_1^2/P$ by suppressing the norm of ${\bf H}[k]$.
Hence, we omit the fourth constraint and transform the JCAS optimization problem of \eqref{MICONs} into
\begin{align}\label{joi}
 {\rm \arg} &\, \mathop{\min}\limits_{\pr[k]} - J<0\notag\\
{\rm s.t.} &\|\pr[k]\|_F^2\leq U,\notag\\
 &\|\pr[k]- \pr_{\rm MI}^\star[k]\|_F^2\leq \rho,
\end{align}
where $\pr_{\rm  MI}^\star[k]=\hs[k]{\bf\Lambda}[k]$.
Note that all variables are subcarrier dependent. For notational simplicity, we omit the parameter $k$ in the following of this subsection.

We define the two constraint functions in \eqref{joi} as $\psi(\pr)=\|\pr\|_F^2=\sum\limits_{u=1}^U\|\pu_u\|_F^2$ and $\psi'(\pr)=\|\pr- \pr_{\rm rad,MI}^\star \|_F^2=\sum\limits_{u=1}^U \|\pu_u - \hs\lambda_u\|_F^2$, where $\lambda_u$ is the $u$th diagonal entry of ${\bf \Lambda}$.

Before optimizing $\pr$, we note that ${\bf \Lambda}$ in $\psi'(\pr)$ is still undetermined.
It is clear to see that $\psi(\pr)=U$ and $\psi'(\pr)=\rho$ are two complex spheres with the sphere center being ${\bf 0}$ and $\hs{\bf\Lambda}$, respectively. The objective function of $J$ can be seen as a saddle surface with the saddle point being the original point. The gradient of $J(\pr)$ is given by
\begin{align}
\frac{\partial J}{\partial \pu_u}=2\rx_u\pu_u.
\end{align}
We determine ${\bf \Lambda}$, such that the absolute value of gradient is maximized, i.e.,
\begin{align}\label{LAM}
 {\rm \arg} &\, \mathop{\max}\limits_{\bf\Lambda}  \sum\limits_{u=1}^U\|2\rx_u\hs\lambda_u\|_F^2\notag\\
{\rm s.t.} &  \|{\bf \Lambda} \|_F^2=\frac{U}{\| \hs \|_F^2}.
\end{align}
By using Least Square method, ${\bf \Lambda}$ is obtained as $\lambda_u^\star= c\|\rx_u\hs\|_F$, where $c$ is a scaling factor, such that $\|{\bf \Lambda}^\star \|_F^2=\frac{U}{\| \hs \|_F^2}$.

The objective function in \eqref{joi} is a quadratic optimization problem with convex constrained sets. Such problems can be iteratively solved by using Newton's method \cite{Newton}. Besides, we also obtain the closed-form solution, as demonstrated in Theorem \ref{T2} and Corollary \ref{C1}.
\begin{theorem}\label{T2}
The minimal value of $-J$ in \eqref{joi} is achieved when $\pr$ is on the surface of the constrained sets.
\end{theorem}
\begin{proof}
The proof is shown in Appendix \ref{AppCC}.
\end{proof}

According to Theorem \ref{T2}, we obtain an important corollary that can directly obtain the closed-form solution of $\pr$, which is illustrated as follows.
\begin{corollary}\label{C1}
When there exist non-zero scaling factors, $f_u$, such that $\psi(f_u\bar{\bf v}_{R,u}) = U \bigcap \psi'(f_u\bar{\bf v}_{u,R})\leq \rho  $ or  $\psi(f_u\bar{\bf v}_{u,R}) \leq U \bigcap \psi'(f_u\bar{\bf v}_{u,R})= \rho $, the closed-formed solution of $\pr$ is obtained as ${\pu}_u^\star =f_u\bar{\bf v}_{u,R}$.
\end{corollary}
\begin{proof}
The proof is shown in Appendix \ref{App2}.
\end{proof}

However, the closed-form solution exists only when $f_u$ exists. In more general cases, we can obtain the following corollary.
\begin{corollary}\label{C2}
The optimal point of $\pr$ is on the tangent plane of $-J(\pr)=J_0$ with $J_0$ being the minimal value.
\end{corollary}
\begin{proof}
The proof is shown in Appendix \ref{App3}.
\end{proof}

\subsection{Iterative Solution}
Next, we propose an iterative algorithm based on Newton's method \cite{Newton}, when the closed-form solution cannot be obtained.
According to the proof of Theorem \ref{T2}, we can make $\pu_u $ reach one surface of the constraints.  When $\pu_u$ is on the surface of constraints, we aim to make the objective function be further decreased. Since there are two surfaces and only one will be reached, we consider two cases as follows.

\subsubsection{Case 1: $\psi(\pr)< U$ and $\psi'(\pr)=\rho$}
This case refers to the case when the surface of $\psi'(\pr)=\rho$ is first reached.
The current iterative real-valued $\pr$ is denoted as $\bar\pr^{(i)}$ and each column of $\bar\pr^{(i)}$ is $\bar\pu_u^{(i)}$. We obtain two new columns moving in the directions of $\bar{\bf v}_{u,R}$, with $\bar{\bf v}_{u,R}$ being defined in \eqref{threeterm},
\begin{align}\label{n1}
& \bar\pu_{u-}^{(i)} =  \bar\pu_u^{(i)}-\epsilon \bar{\bf v}_{u,R},\notag\\
& \bar\pu_{u+}^{(i)}=  \bar\pu_u^{(i)}+\epsilon \bar{\bf v}_{u,R}.
\end{align}
Note that neither $\bar\pu_{u+}^{(i)}$ nor $\bar\pu_{u-}^{(i)}$  satisfies that $\psi'(\pr)=\rho$. Hence, we need to scale these columns onto the surface of $\psi'(\pr)=\rho$, i.e.,
\begin{align}\label{newq1}
 \bar\pu_{u-}^{(i+1)}=\bar\pu_{u-}^{(i)} + a_{u-} (\bar {\bf c}_{u}-\bar\pu_{u-}^{(i)}),\notag\\
 \bar\pu_{u+}^{(i+1)}=\bar\pu_{u+}^{(i)} + a_{u+} (\bar {\bf c}_{u}- \bar\pu_{u+}^{(i)}),
 \end{align}
where $\bar{\bf c}_u=[\rep[\hs\lambda_u]^T,\imp[\hs\lambda_u]]^T$, $a_{u-}$ and $a_{u+}$ are scaling coefficients, such that $\psi'\left(\bar\pu_{u-}^{(i+1)}\right)=\rho$ and $\psi'\left(\bar\pu_{u+}^{(i+1)}\right)=\rho$, respectively.

The columns defined in \eqref{newq1} should make the objective function keep dropping. We select the one of which objective function is reduced, i.e.,
\begin{align}\label{comp}
\bar\pu_u^{(i+1)}=\left\{\begin{array}{cc}
\bar\pu_{u-}^{(i+1)}& {\rm if} \, -J(\bar\pu_{u-}^{(i+1)})<-J(\bar\pu_{u+}^{(i+1)})\\
\bar\pu_{u+}^{(i+1)}& {\rm else}\\
\end{array}
\right.
\end{align}

\begin{algorithm}[t]
	\caption{JCAS Precoder Optimization with MI Constraint}\label{Alg1}
	\begin{algorithmic}[1]
    \STATE {\bf Input: } $\hs$ and $\{\rx_u\}_{u=1}^U$.

    \STATE {\bf Initialization: }  $i=0$, $J(\pr^{(-1)})=-\infty$, $\epsilon$, $\rho$,

    \STATE Obtaining ${\bf v}_{u,R}$ that is the $R$th eigenvector of $\rx_u$, with $R$th eigenvalue being negative.

     \STATE Finding $\pr^{(i)}$ satisfying Case 1 or Case 2.
    \STATE {\textbf{Case 1:}}
    \WHILE {$-J(\pr^{(i-1)})>-J(\pr^{(i)})$ and $\psi(\pr^{(i)})< U$ and $\psi'(\pr^{(i)})=\rho$}
    \STATE Obtaining $\pu_{u\pm}^{(i)}$ according to \eqref{n1}.
    \STATE Scaling $\pu_{u\pm}^{(i)}$ and obtaining $\pu_{u\pm}^{(i+1)}$ according to \eqref{newq1}.
    \STATE Obtaining next iterative vectors, $\pu_{u}^{(i+1)}$,  according to \eqref{comp}.
    \STATE $i=i+1$.
    \ENDWHILE

    \STATE {\textbf{Case 2:}}
    \WHILE {$-J(\pr^{(i-1)})>-J(\pr^{(i)})$ and $\psi(\pr^{(i)})=U$ and $\psi'(\pr^{(i)})<\rho$}
    \STATE Obtaining $\pu_{u\pm}^{(i)}$ according to \eqref{n1}.
    \STATE Scaling $\pu_{u\pm}^{(i)}$ and obtaining $\pu_{u\pm}^{(i+1)}$  according to \eqref{newq2}.
    \STATE Obtaining next iterative vectors, $\pu_{u}^{(i+1)}$,  according to \eqref{comp}.
    \STATE $i=i+1$.
    \ENDWHILE

    \STATE {\bf Output:}  $\pr$.
	\end{algorithmic}
\end{algorithm}

Updating the iteration index $i=i+1$ and repeat the same procedure. We terminate the iteration when $-J(\bar\pu_u^{(i+1)})$ stop dropping or when the other constraint is not satisfied, i.e., $\psi(\pr)> U.$

\subsubsection{Case 2: $\psi(\pr)=U$ and $\psi'(\pr)<\rho$}

In Case 2, we adopt the same strategy as Case 1. In this case, two temporary vectors are still given by \eqref{n1}.
Then, we scale them onto the surface of $\psi(\pr)=U$, i.e.,
\begin{align}\label{newq2}
 \bar\pu_{u-}^{(i+1)}=  b_{u-}  \bar\pu_{u-}^{(i)},\notag\\
\bar\pu_{u+}^{(i+1)}=  b_{u+}  \bar\pu_{u+}^{(i)},
\end{align}
where $b_{u,\pm}$ are scaling coefficients, such that $\psi\left(\pu_{u\pm}^{(i+1)}\right)=U$.

Likewise, we select the one of which objective function is less as in \eqref{comp}.
Then we update the iteration index $i=i+i$ and repeat the same procedure. We terminate the iteration $-J(\bar\pu_u^{(i+1)})$ stop dropping or when the other constraint is not satisfied, i.e., $\psi'(\pr)>\rho.$

\subsection{Complexity Analysis}
In this subsection, we analyze the computational complexity of Algorithm 1. The main computation tasks in Algorithm 1 include the eigen-decomposition of ${\bf R}_u$ and the iterations. The complexity of the eigen-decomposition of ${\bf R}_u$ can be given by $\mathcal O(N^3)$. Since there are $UK$ ${\bf R}_u[k]$, the total complexity of eigen-decomposition  is $\mathcal O(N^3UK)$. The main steps in the iterations are Step 8 for Case 1 and Step 15 for Case 2, which both have a complexity of $\mathcal O(N)$, due to the calculation of $a_{u{\pm}}$ and ${b_{u\pm}}$. Hence, the complexity of the iterations is $\mathcal O(NUKN_{it})$, where $N_{\rm it}$ is the number of iterations.  The overall complexity can be given by  $\max\left(\mathcal O(N^3UK),\mathcal O(NUKN_{it})\right)$.

\section{CRB Constrained JCAS Waveform Optimization}

In this section, we optimize the proposed communication metric,
$J$, with constraining the CRB of the radar. Intuitively, we would like to adopt the same strategy as in MI constrained problem, which controls the Euclidean distance between the precoder and the optimal individual radar precoder. Nevertheless, for individual CRB optimization, we obtain the optimal covariance  matrix, denoted as ${\bf Q}^\star_{\rm CRB}[k]$, instead of the optimal $\pr[k]$. With a given ${\bf Q}[k]=\pr[k]^H\pr[k]$, there are infinite solutions of $\pr[k]$. Not all the solutions are appropriate ones for the communication metric. Using the Euclidean distance between ${\bf Q}[k]$ and ${\bf Q}^\star_{\rm CRB}[k]$ would make the variables have an order of four and hence difficult to solve.  Aiming to reduce the order of variables, we constrain the CRB by controlling
\begin{align}\label{CRBcons}
&\psi''(\pr[k])\notag\\
=&\left|\rep[\pr^H[k]\pr[k]-{\bf Q}_{\rm CRB}^\star[k]]+\imp[\pr^H[k]\pr[k]-{\bf Q}^\star_{\rm CRB}[k]]\right|\notag\\
\leq& \xi,
\end{align}
where $\xi$ is the threshold to control the CRB. This constraint has an order of two for the variables in $\pr[k]$.

Still, we omit the subcarrier $k$ in the following of this section. The JCAS optimization problem with the CRB constraint is formulated as
 \begin{align}\label{CRBCONs}
{\rm \arg}\,& \mathop{\min}\limits_{\pr}  -J <0 \notag\\
{\rm s.t.} & \psi(\pr)\leq U,\notag\\
& \psi''(\pr )\leq \xi.
\end{align}
Similar to the MI constrained problem, the transformed CRB constrained problem can also be seen as  a quadratic problem with convex constraints. Such problems can be solved using Newton's method.

Referring to Corollary \ref{C1}, the closed-form solution can be obtained when $f_u\bar{\bf v}_{u,R}$ can reach the surfaces of $\psi(\pr)=U$ or $\psi''(\pr)=\xi$. When the closed-form solution does not exist, we propose an iterative algorithm as in Algorithm 1.

\subsubsection{Case 3: $\psi(\pr)< U$ and $\psi''(\pr)= \xi$}
In this case,  two temporary vectors are still given by
$ \bar\pu_{u-}^{(i)}$ and $ \bar\pu_{u+}^{(i)}$ of which expressions are given by \eqref{n1}.
Then, we scale them onto the surface of $\psi''(\pr)= \xi$. The new columns are given by
\begin{align}\label{newq3}
& \bar\pu_{u-}^{(i+1)} =  c_{u-} \bar\pu_{u-}^{(i)},\notag\\
& \bar\pu_{u+}^{(i+1)}=  c_{u+} \bar\pu_{u+}^{(i)},
\end{align}
where $c_{u\pm}$ are scaling coefficients, such that $\psi''(\pu_{u\pm}^{(i+1)})=\xi$.

Likewise, we select the one of which objective function is less as in \eqref{comp}.
Then we update the iteration index $i=i+i$ and repeat the same procedure. We terminate the iteration $-J(\bar\pu_u^{(i+1)})$ stop dropping or when the other constraint is not satisfied, i.e., $\psi(\pr)>U.$

\begin{algorithm}[t]
	\caption{JCAS Precoder Optimization with CRB Constraint}\label{Alg2}
	\begin{algorithmic}[1]
    \STATE {\bf Input: } ${\bf Q}^\star[k]$ and $\{\rx_u\}_{u=1}^U$.

    \STATE {\bf Initialization: }  $i=0$, $J(\pr^{(-1)})=-\infty$, $\epsilon$, $\xi$,

    \STATE Obtaining ${\bf v}_{u,R}$ that is the $R$th eigenvector of $\rx_u$, with $R$th eigenvalue being negative.

     \STATE Finding $\pr^{(i)}$ satisfying Case 3 or Case 4.
    \STATE {\textbf{Case 3:}}
    \WHILE {$-J(\pr^{(i-1)})>-J(\pr^{(i)})$ and $\psi(\pr^{(i)})< U$ and $\psi''(\pr^{(i)})=\xi$}
    \STATE Obtaining $\pu_{u\pm}^{(i)}$.
    \STATE Scaling $\pu_{u\pm}^{(i)}$ and obtaining $\pu_{u\pm}^{(i+1)}$ according to \eqref{newq3}.
    \STATE Obtaining next iterative vectors, $\pu_{u}^{(i+1)}$,  according to \eqref{comp}.
    \STATE $i=i+1$.
    \ENDWHILE

    \STATE {\textbf{Case 4:}}
    \WHILE {$-J(\pr^{(i-1)})>-J(\pr^{(i)})$ and $\psi(\pr^{(i)})=U$ and $\psi''(\pr^{(i)})<\xi$}
    \STATE Obtaining $\pu_{u\pm}^{(i)}$.
    \STATE Scaling $\pu_{u\pm}^{(i)}$ and obtaining $\pu_{u\pm}^{(i+1)}$  according to \eqref{newq4}.
    \STATE Obtaining next iterative vectors, $\pu_{u}^{(i+1)}$,   according to \eqref{comp}.
    \STATE $i=i+1$.
    \ENDWHILE

    \STATE {\bf Output:}  $\pr$.
	\end{algorithmic}
\end{algorithm}

\subsubsection{Case 4: $\psi(\pr)=U$ and $\psi''(\pr)<\xi$}

In Case 4, two temporary vectors are given by $\bar\pu_{u-}^{(i)}$  and $ \bar\pu_{u+}^{(i)}$.
Then, we scale them onto the surface of $\psi(\pr)=U$, i.e.,
\begin{align}\label{newq4}
 \bar\pu_{u-}^{(i+1)}=  d_{u-}  \bar\pu_{u-}^{(i)},\notag\\
\bar\pu_{u+}^{(i+1)}=  d_{u+}  \bar\pu_{u+}^{(i)},
\end{align}
where $d_{u,\pm}$ are scaling coefficients, such that $\psi\left(\pu_{u\pm}^{(i+1)}\right)=U$.

Likewise, we select the one of which objective function is less as in \eqref{comp}.
Then, we update the iteration index $i=i+i$ and repeat the same procedure. We terminate the iteration when $-J(\bar\pu_u^{(i)})$ stop dropping or when the other constraint is not satisfied, i.e., $\psi''(\pr)>\xi.$
The proposed CRB constrained optimization scheme is summarized in Algorithm \ref{Alg2}.
Note that from Case 1 to Case 4, all procedures are the same except the determination of the scaling coefficients and the terminating conditions.

\section{Simulation Results}
\begin{figure}[t]
    \centering
    \includegraphics[scale=0.8]{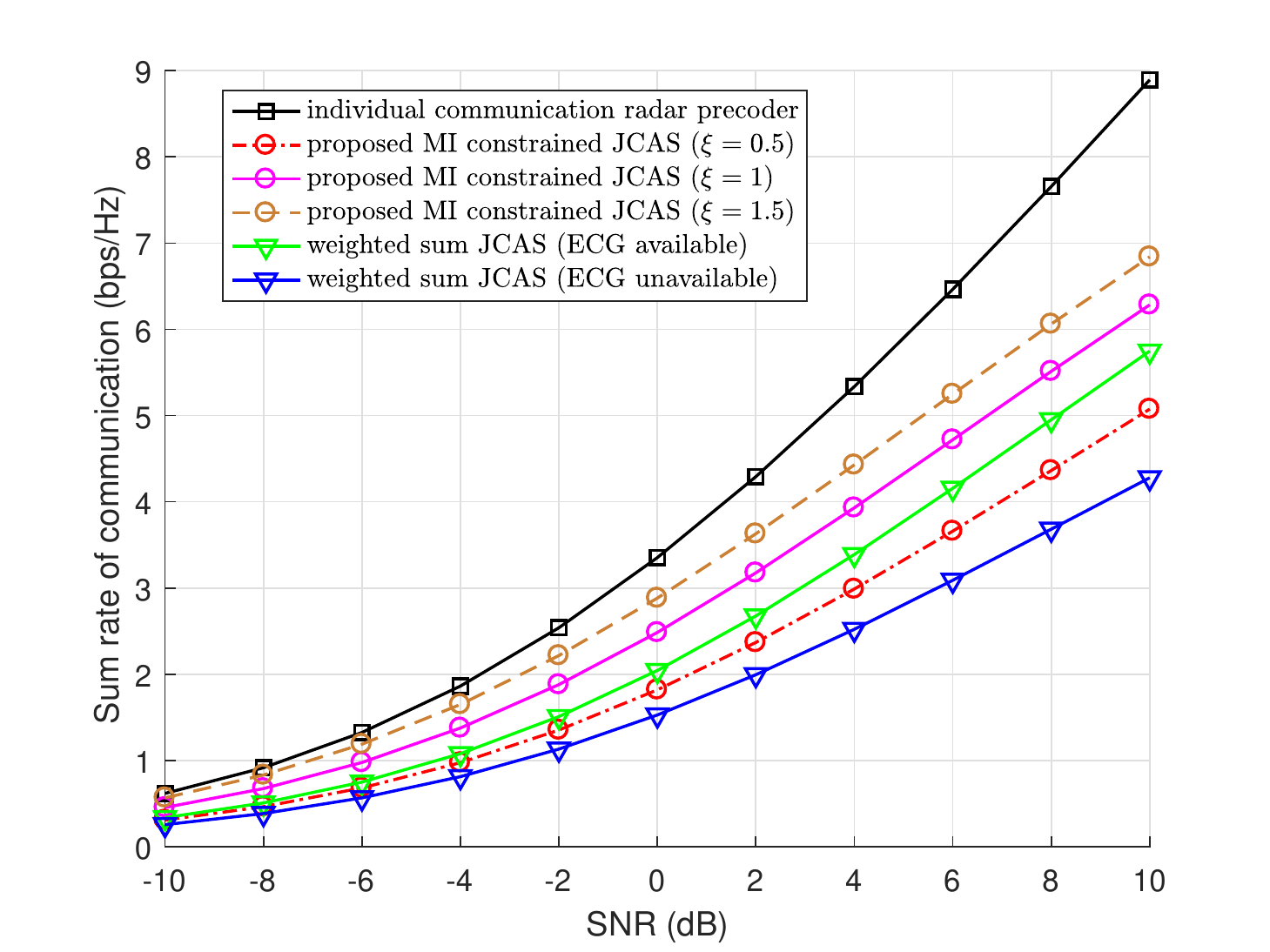}
    \caption{  Sum rates of communications versus SNR with using Algorithm 1 and other bench-marking JCAS solutions.}
    \label{Fig_1}
\end{figure}

\begin{figure}[t]
    \centering
    \includegraphics[scale=0.8]{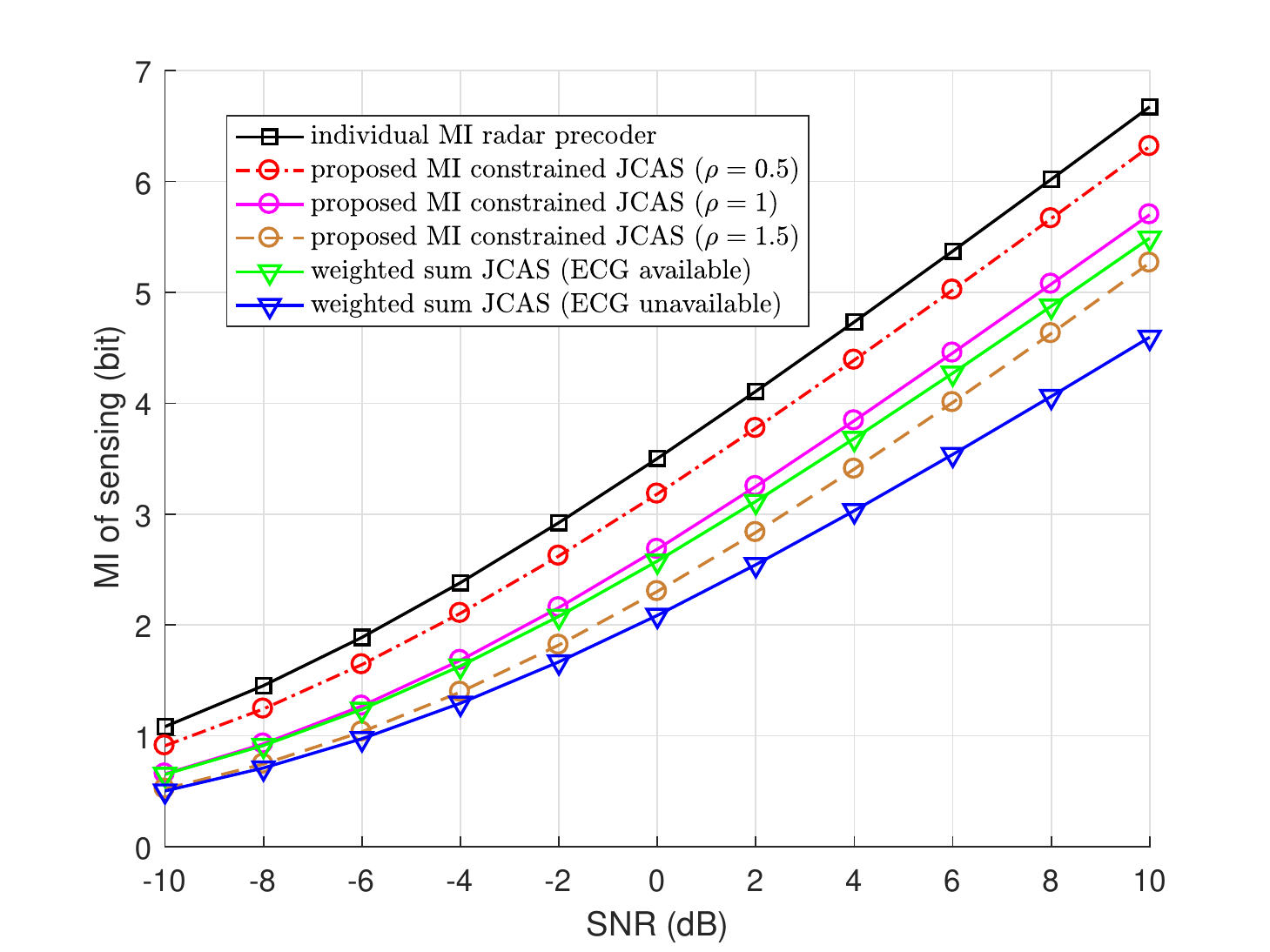}
    \caption{MI of sensing versus SNR with using Algorithm 1 and other bench-marking JCAS solutions.}
    \label{Fig_2}
\end{figure}
\begin{figure}[t]
    \centering
    \includegraphics[scale=0.8]{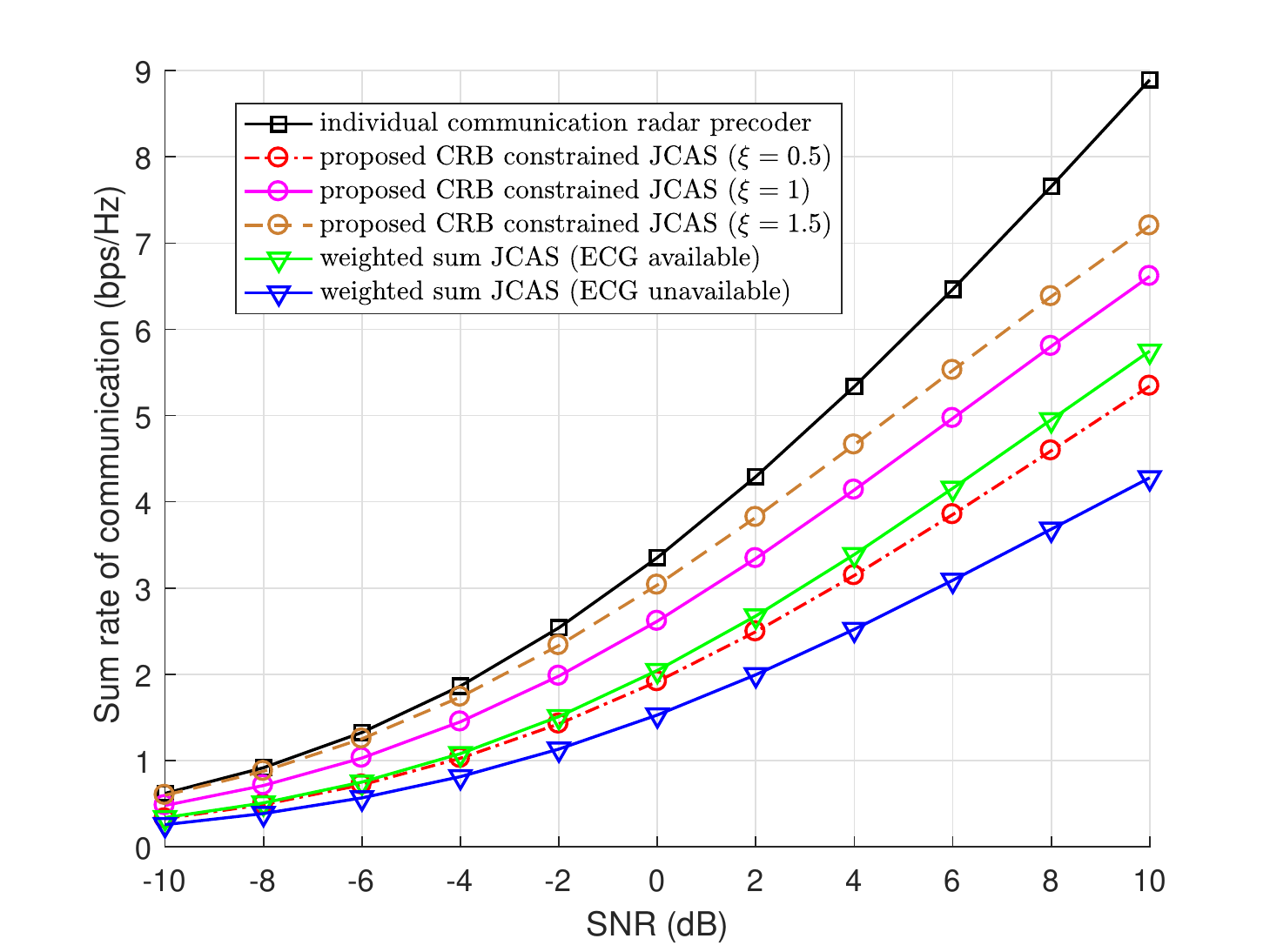}
    \caption{ Sum rates of communications versus SNR with using Algorithm 2 and other bench-marking solutions. }
    \label{Fig_3}
\end{figure}

\begin{figure}[t]
    \centering
    \includegraphics[scale=0.8]{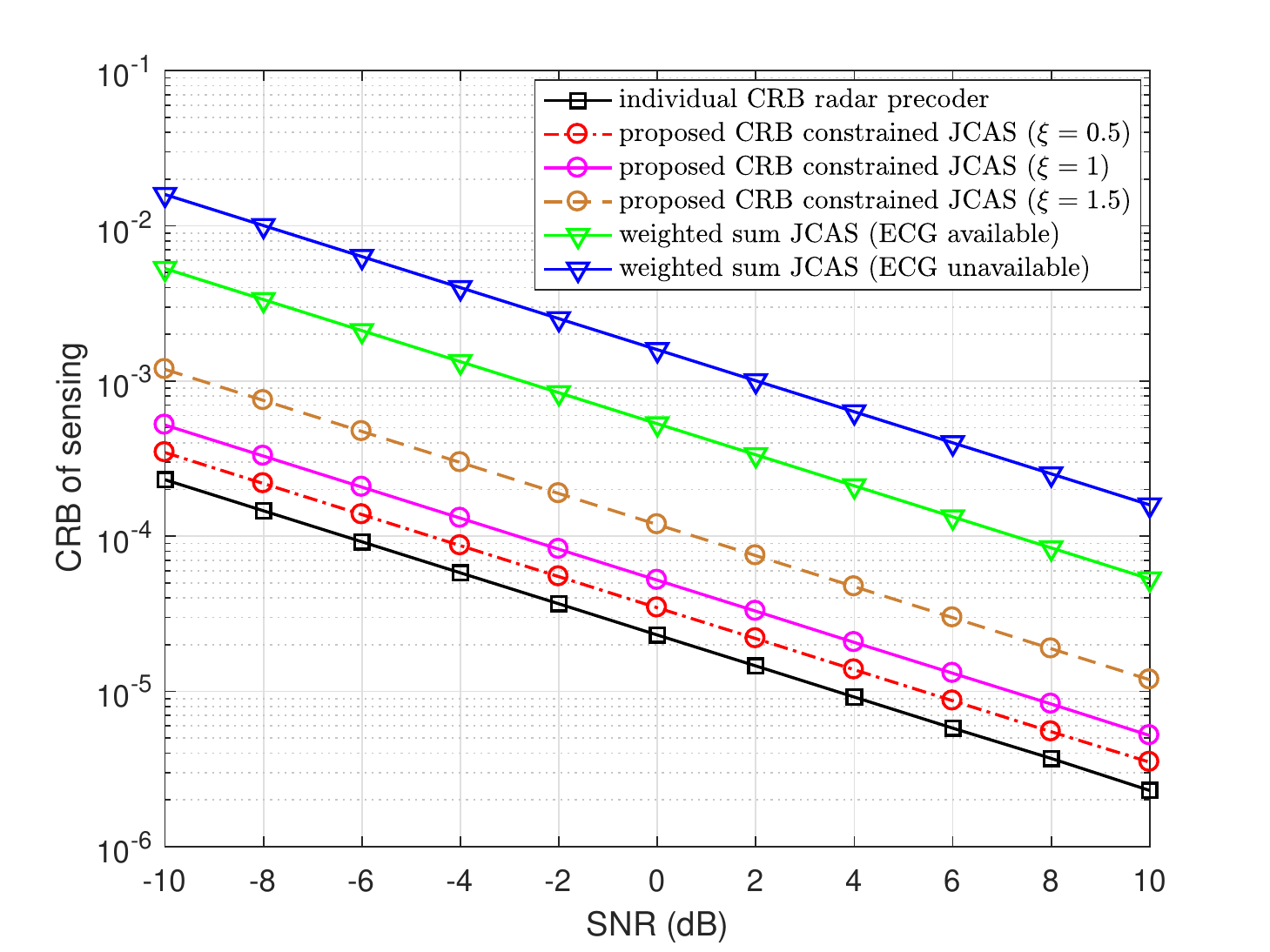}
    \caption{CRB of sensing versus SNR with using Algorithm 2 and other bench-marking solutions.}
    \label{Fig_4}
\end{figure}
 \begin{figure}[t]
    \centering
    \includegraphics[scale=0.8]{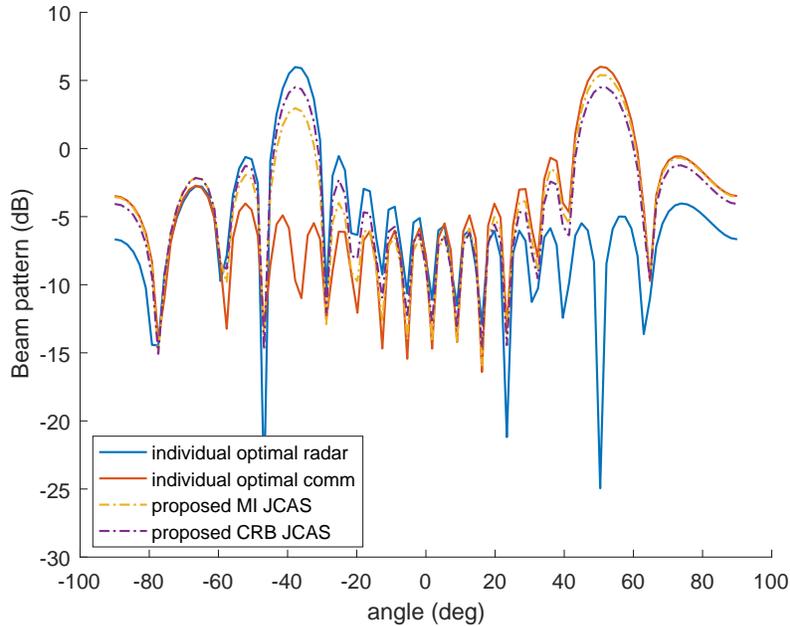}
    \caption{Beam patterns with employing individual and JCAS solutions.}
    \label{Fig_5}
\end{figure}

\begin{figure}[t]
    \centering
    \includegraphics[scale=0.8]{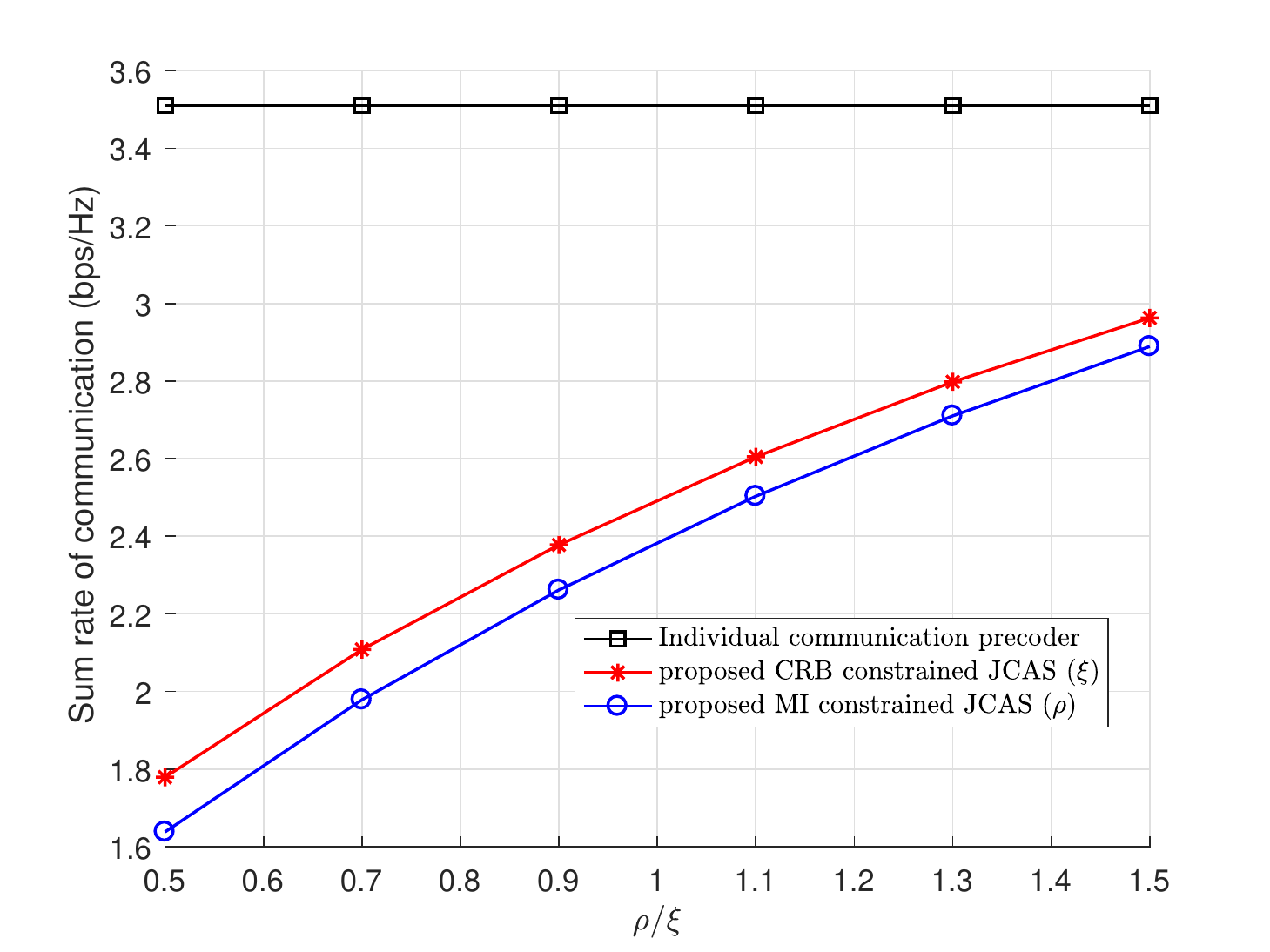}
    \caption{ Sum rate versus $\rho/\xi$  for communications with employing individual and JCAS solutions.}
    \label{Fig_6}
\end{figure}

\begin{figure}[t]
    \centering
    \includegraphics[scale=0.8]{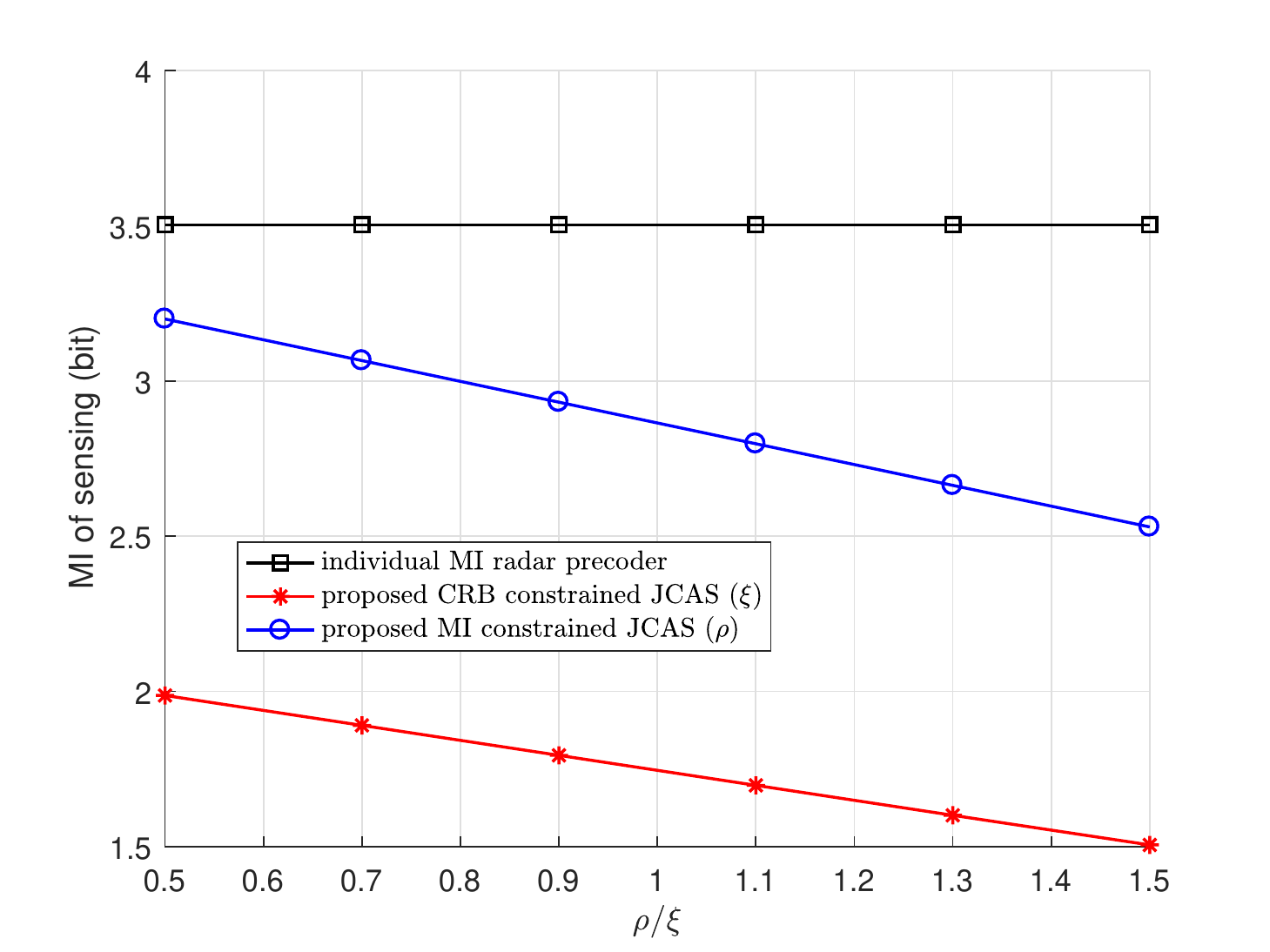}
    \caption{MI versus $\rho/\xi$  for sensing with employing individual and JCAS solutions.}
    \label{Fig_7}
\end{figure}

\begin{figure}[!ht]
    \centering
    \includegraphics[scale=0.8]{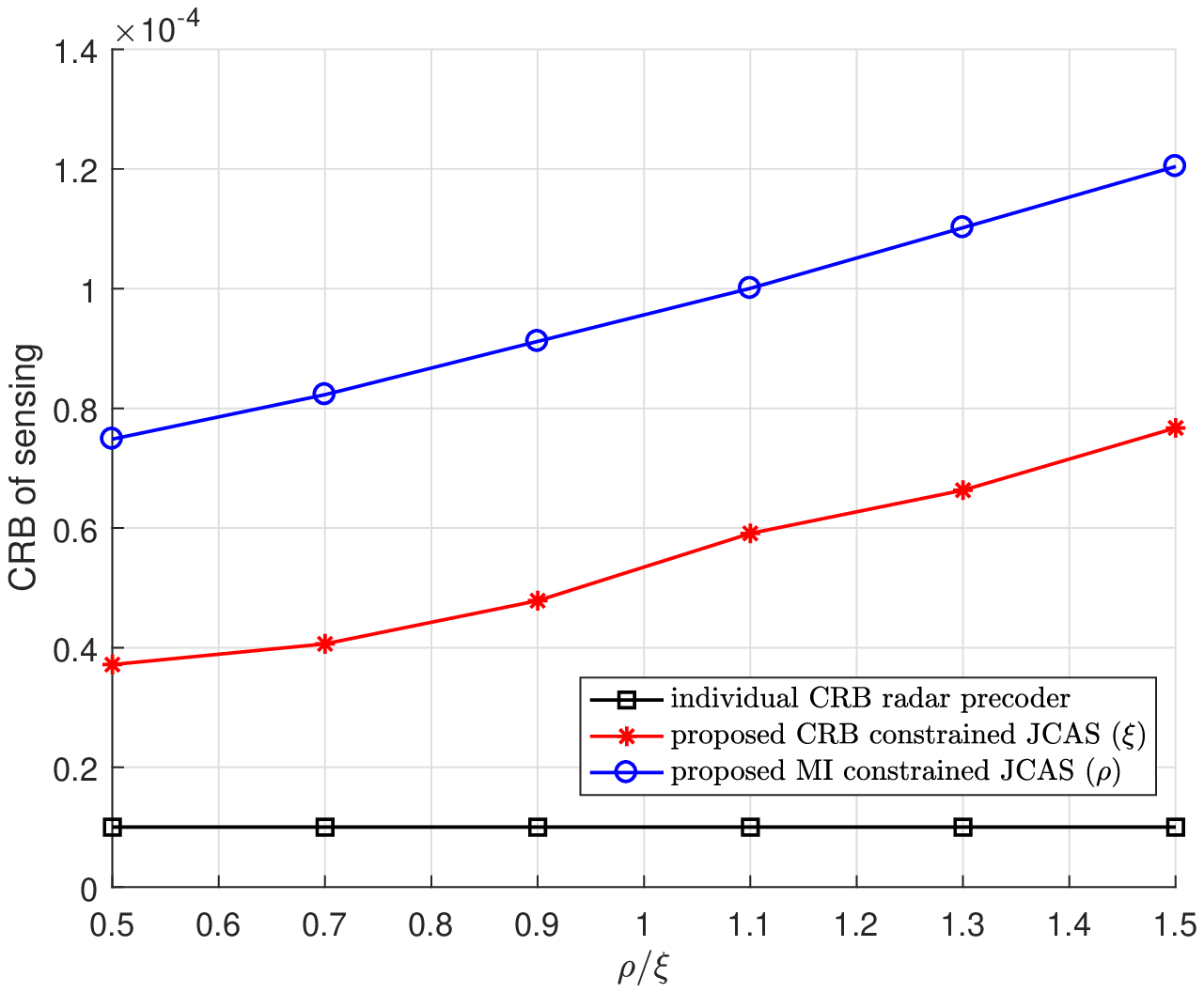}
    \caption{CRB versus $\rho/\xi$  for sensing with employing individual and JCAS solutions.}
    \label{Fig_8}
\end{figure}
In this section, we provide simulation results to validate the proposed algorithms, following the system and channel models introduced in Section II.  We simulate a JCAS system where a BS communicates with $U=2$ MUs each having $L_u=3$ paths and detects $L=3$ targets. The BS adopts a $16\times 1$ ULA as the transmit antennas and uses OFDM modulation to transmit a $2\times 1$ data symbol vector on each subcarrier. The number of subcarriers is $K=512$. The AoDs for both MUs and targets are randomly distributed from $-\pi$ to $\pi$.  The delay of each path is a random value ranging from 0 to $T_{\rm C}=0.1$ ms. The period of one OFDM symbol is $T=0.2$ ms.  For the parameter in the communication metric, we let $\mu=5$. For parameters in the proposed algorithms, we let $\epsilon=0.05$, $\mu\in[0.5, 1.5]$, and $\xi\in[0.5,1.5]$.

We separate the simulation results into three subsections. In the following subsection \ref{MI-JCAS}, we  show the achieved  sum rates for communications and the achieved MI for radar sensing by using our proposed Algorithm 1. In subsection \ref{CRB-JCAS}, we  verify the sum rates for communications and the CRB for radar sensing by using the proposed Algorithm 2. For comparison, we compare our proposed algorithms with the bench-marking solutions, including the optimal individual solutions and the weighted sum JCAS solution that is proposed in \cite{liuweight}. In subsection \ref{mutual-JCAS}, we aim to compare our proposed algorithms in parallel and figure out the impacts of using MI and CRB on optimizing JCAS waveforms.

\subsection{MI Constrained JCAS}\label{MI-JCAS}

Fig. \ref{Fig_1} shows how the sum rates vary with SNR using different JCAS solutions.  For the individual communication precoder, the sum rate remains the highest. It is undoubted since the precoder is optimized solely for maximizing the sum rate. Regarding the JCAS solutions, we compare our proposed Algorithm 1 with the weighted-sum JCAS solution in \cite{liuweight}. The weighted-sum JCAS solution needs the prior knowledge about the ECG. We see that with $\rho$ ranging from 1 to 1.5, our proposed precoder achieves better performance than the weighted-sum JCAS solution. Moreover, our proposed algorithms do not need to know the value of ECG. When ECG is unavailable, we see that the sum rate of the weighted-sum JCAS solution drops substantially. Our proposed algorithm guarantees that the ratio of ECG to MUI is larger than $\mu$. We can obtain the sum rate that is no smaller than $\mu$ at high SNRs.   The MUI has a dominant impact on the sum rate when SNR is high, hence our proposed precoder can achieve better performance by further increasing $\mu$.

Fig. \ref{Fig_2} unfolds the achieved MI of radar sensing versus SNR using different JCAS solutions. The system setup is the same as that in Fig. \ref{Fig_1}. The MI of individual MI radar precoder, i.e., $\pr[k]=\hs[k]{\bf \Lambda}[k]$, remains the highest.  Our proposed Algorithm 1 achieves better MI than that of \cite{liuweight} when $\rho$ is no larger than 1. Without obtaining ECG, the MI of weight-sum JCAS solution declines greatly. By decreasing $\rho$ to zero, our proposed solution can achieve the MI of individual MI radar precoder.  Together with Fig. \ref{Fig_1}, under the same power control, we notice that the sum rate is in proportion to $\rho$ while the MI is in inverse proportion to $\rho$. To guarantee both good performances for communications and sensing, $\rho$ should be nearly 1.

\subsection{CRB Constrained JCAS}\label{CRB-JCAS}

Fig. \ref{Fig_3} depicts how the sum rates vary with SNR using Algorithm 2 and other solutions. Similar to Fig. \ref{Fig_1},  the sum rate of the individual communication precoder remains the highest. Regarding the JCAS solutions, we compare our proposed Algorithm 2 with  the weighted-sum JCAS approach in \cite{liuweight}. We observe that, with $\xi$ ranging from 1 to 1.5, our proposed Algorithm 2 achieves better performance than the weighted-sum JCAS solution. Comparing with Fig. \ref{Fig_1}, We note that under the same control of $\rho (\xi)$, the achieved sum rate of Algorithm 2 is slightly better than that of Algorithm 1. We will discuss the gap later in Fig. \ref{Fig_6}.

Fig. \ref{Fig_4} illustrates the achieved CRB of radar sensing versus SNR using different JCAS solutions. The system setup is the same as that in Fig. \ref{Fig_1}. The individual CRB radar precoder achieves the lowest CRB among all solutions.  Our proposed Algorithm 2 achieves lower CRB than that of \cite{liuweight} no matter the ECG is available or not. By decreasing $\xi$ to zero, our proposed solution can achieve the optimal lowest CRB.  Together with Fig. \ref{Fig_3}, we can let $\xi$ range from $1$ to $1.5$ to guarantee both good performances for communication and sensing.
\subsection{Comparisons of Proposed MI/CRB Constrained JCAS}\label{mutual-JCAS}

Fig. \ref{Fig_5} illustrates an example of JCAS beam patterns with employing the proposed algorithms. To make the beam patterns clear to see, we only consider one MU and one target with only one path, i.e.,  $U=1$,   $L=1$, and $L_u=1$. For individual designs as illustrated in section III, we adopt the individual optimal MI radar precoder and the individual optimal communication precoder, respectively. We see that there is one main lobe located at around $-35\degree$ for the radar precoder, and one main lobe located at around $45\degree$ for the communication precoder. For JCAS designs, both our proposed algorithms have two main lobes, matching with the locations of individual designs tightly.   At the sidelobes, we see that our proposed JCAS solutions also have a good suppression, especially at the range of $(-30\degree, 30\degree).$

Next, we compare the proposed Algorithm 1 with Algorithm 2 under equal communication performance. This can unfold the impact of CRB and MI on the JCAS designs. To guarantee that Algorithm 1 and 2 achieve equal sum rates for communications, we simulate the sum rate versus threshold factors, ($\rho/\xi$), as shown in Fig. \ref{Fig_6}. The system set up is the same as that in Fig. \ref{Fig_1}, except that the SNR is fixed at 0 dB. It is clear to see that, under the equal value of the threshold factors, Algorithm 2 achieves slightly better sum rates that Algorithm 1. When two proposed algorithms achieve equal sum rates, $\rho$ needs to be about 0.11 larger than $\xi$. This result helps us to analyze the following two results.

In Fig. \ref{Fig_7}, we aim to verify if Algorithm 2 can achieve effective MI. The system set up is the same as that in Fig. \ref{Fig_6}. Note that Algorithm 2 is designed for minimizing CRB.  The MI of individual MI radar precoder remains the highest.  Referring to the result in Fig. \ref{Fig_6}, we note that $\rho$ should be about 0.1 larger than $\xi$ to make the two proposed algorithms achieve an equal level of sum rates for communications. We see that the MI achieved by Algorithm 2 is significantly lower than that of Algorithm 1. With the same communication performances, the MI of Algorithm 2  is less that than the optimal solution by almost 2 bits. This could be explained by noting that the individual CRB radar precoder is close to the array response matrix, $\bf A$, instead of $\hs[k]$, which means that one column could be truncated by using the CRB precoder and results in the decline of MI for using the CRB constrained precoder.

In Fig. \ref{Fig_8}, we aim to verify if Algorithm 1 can achieve effective CRB. The system set up is the same as that in Fig. \ref{Fig_6}. The CRB of individual CRB radar precoder remains the lowest. We observe that the CRB of Algorithm 1 is almost equal with that of Algorithm 2 and rises slightly with increasing $\rho/\xi$. When $\rho$ is 0.1 larger than $\xi$, the gap between the two algorithms is around $2\times 10^{-5}$. Even though the CRB of MI constrained algorithm is higher than that of the CRB constrained algorithm, the gap of $2\times 10^{-5}$ can be neglected. Referring to Fig. \ref{Fig_7}, we see that the MI constrained optimization can achieve good performances on both MI and CRB. More importantly, the MI constrained optimization has a closed-form solution while the CRB constrained optimization has not. Hence, we can obtain a conclusion that maximizing MI is more efficient and simpler than minimizing CRB for JCAS optimization.

\section{Conclusion}\label{sec-conc}

In this paper, we have proposed the JCAS waveform optimization algorithms that maximize the developed lower bound of communication sum rates under the constraint of either MI or CRB of radar.  The adopted radar metrics correspond to two types of individual optimal radar matrices, i.e., precoding matrix and covariance matrix. The developed lower bound of communication rates balances between MUI and ECG, and serves as a good substitute for SINR. We have established the connections between different constrained optimizations and showed that MI constrained JCAS solution is more efficient and less complicated than the CRB constrained JCAS solution.

\begin{appendices}
\section{Proof for Theorem \ref{T1}}\label{App0}
Given that $I_{k,u}+U\sigma_1^2/P\leq1\leq\frac{S_{k,u}}{S_{k,u}-\mu I_{k,u}}$  and $S_{k,u}-\mu I_{k,u}>0$, we have $S_{k,u}-\mu I_{k,u}\leq \frac{S_{k,u}}{I_{k,u}+U\sigma_1^2/P}$. Then, we have
\begin{align*}
 {\rm SINR}_{k,u}=&\frac{\frac PU|(\hc[k])^H{\pu_u[k]}|^2}{\frac PU \sum\limits_{v\neq u}|({\bf h}^{\rm C}_v[k])^H{\pu_u[k]}|^2+\sigma_1^2}\notag\\
             =&\frac{PS_{k,u}/U}{PI_{k,u}/U+\sigma_1^2}\notag\\
             =&\frac{ S_{k,u} }{ I_{k,u} +U\sigma_1^2/P}\notag\\
             \geq&S_{k,u}-\mu I_{k,u}\triangleq  J_{k,u},
\end{align*}
We note that the SINR is the sum of ${\rm SINR}_{k,u}$ and $J$ is the sum of $J_{k,u}$. Therefore, $J$ is a lower bound of SINR.
\section{Expression for FIM Subblocks}\label{App1}
Referring to the derivation in \cite{perfoB}, with including the precoding matrix, $\pr[k]$, the FIM subblocks are expressed as
\begin{align}
{\bf F}_1&=\frac{2}\sigma\sum\limits_{k=1}^K {\bf W}^H[k]{\bf A}^H\pr[k]\pr^H[k]{\bf A}{\bf W}[k]\notag\\
{\bf F}_2&=\frac{2}\sigma\sum\limits_{k=1}^K {\bf W}^H[k] ({\bf A}^H\pr[k]\pr^H[k]{\bf A})\odot{\bf X} {\dot{\bf W}[k]}\notag\\
{\bf F}_3&=\frac{2}\sigma\sum\limits_{k=1}^K {\bf W}^H[k]\left({\bf A}^H\pr[k]\pr^H[k]\dot{\bf A}  \right)\odot{\bf X}{\bf W}[k]\notag\\
{\bf F}_4&=\frac{2}\sigma\sum\limits_{k=1}^K   \dot{\bf W}^H[k]   ({\bf A}^H\pr[k]\pr^H[k]{\bf A})\odot({\bf X}^H{\bf 1}{\bf 1}^H{\bf X}) \dot{\bf W}[k] \notag\\
{\bf F}_5&=\frac{2}\sigma\sum\limits_{k=1}^K  \dot{\bf W}^H[k] \left({\bf A}^H\pr[k]\pr^H[k]\dot{\bf A} \right)\odot({\bf X}^H{\bf 1}{\bf 1}^H{\bf X}){\bf W}[k]\notag\\
{\bf F}_6&=\frac{2}\sigma\sum\limits_{k=1}^K{\bf W}^H[k] \left( \dot{\bf A}^H\pr[k]\pr^H[k]\dot{\bf A} \right)\odot({\bf X}^H{\bf 1}{\bf 1}^H{\bf X}){\bf W}[k],
\end{align}
where $\odot$ is Hadamard product,
\begin{align}
\dot{\bf W}[k] &=\left[\frac{\partial[{\bf W}[k]]_{:,1}}{\partial \tau_1},\cdots,\frac{\partial[{\bf W}[k]]_{:,L}}{\partial \tau_L}\right]\notag\\
&=-\frac{j2\pi k}T{\bf W}[k],
\end{align}
and
\begin{align}
\dot{\bf A} &=\left[\frac{\partial[{\bf A}]_{:,1}}{\partial \p_1},\cdots,\frac{\partial[{\bf A}]_{:,L}}{\partial \p_L}\right].
\end{align}
\section{Proof for Theorem \ref{T2}}\label{AppCC}
Note that $-J $ is a saddle face that is symmetric with respect to the origin point. We write $\pr$ in the real parts and imaginary parts, i.e.,
\begin{align}
&-J =\sum\limits_{u=1}^U{\rm Re}[\pu_u^H\rx_u\pu_u]\notag\\
&=\sum\limits_{u=1}^U\left[\begin{array}{c}
{\rm Re}[\pu_u]\\
{\rm Im}[\pu_u]\\
\end{array}\right]^H\left[\begin{array}{cc}
{\rm Re}[\rx_u]&{\rm Im}[\rx_u]\\
-{\rm Im}[\rx_u]&{\rm Re}[\rx_u]\\
\end{array}\right]
\left[\begin{array}{c}
{\rm Re}[\pu_u]\\
{\rm Im}[\pu_u]\\
\end{array}\right]\notag\\
&\triangleq\sum\limits_{u=1}^U{\bar\pu}_u^H{\bar\rx}_u{\bar\pu}_u.
\end{align}
Supposed that there is an initial $\{\bar\pu_u^0\}_{u=1}^U$ that satisfies both $\psi(\pr)<U$ and $\psi'(\pr)<\rho$, there exists a new column, $\bar\pu_u^1=\bar\pu_u^0+\epsilon\bar{\bf v}_{u,R}$, where $\bar{\bf v}_{u,R}$ is the $R$th eigen vector of $\bar\rx_u$ with $R$ being the rank of $\rx_u$. The variable, $\epsilon>0$, is small enough to make the new column inside the constraints. Note that $\bar\rx_u$ is not a semi-definite matrix, the $R$th eigen value is negative, denoted as $r_R$, (when all eigen values are sorted in order). Hence, the new column,  $\bar\pu_u^1$, satisfies
\begin{align}\label{threeterm}
-J(\bar\pu_u^1)=-J(\bar\pu_u^0)+2r_R\epsilon(\bar\pu_u^0)^H\bar{\bf v}_{u,R}+\epsilon^2r_R.
\end{align}
where $-J(\bar\pu_u)$ means the value of $-J$ with respect to $\pu_u$.  It is clear that the third term of \eqref{threeterm} is negative.
If the second term is also negative, we have $-J(\bar\pu_u^1)<-J(\bar\pu_0)$. Hence,   $-J$ keeps dropping until $\pr$ is on the surface of the constraints. If the second term in \eqref{threeterm} is positive, then we choose another column, i.e., $\bar\pu_u^{1'}=\bar\pu_u^0-\epsilon\bar{\bf v}_R$, then the second term becomes negative. Therefore, the optimal  $\pr$ is on the surface of the constraints.
\section{Proof for Corollary \ref{C1}}\label{App2}
According to Theorem \ref{T2}, the optimal point of ${\bf P}$ is on the surface of constraints. When $f_u\bar{\bf v}_{u,R}$ is also on the surface of constraints, to further decrease the objective function from $\bar\pu_u^0=f_u\bar{\bf v}_{u,R}$, the possible value of $\bar\pu_u$ is given by
$\bar\pu_u^1=\bar\pu_u^0+\sum_{i=1}^{R-1}\epsilon_i\bar{\bf v}_{u,i}$, where $\bar{\bf v}_{u,i}$ is the $i$th left eigenvector of $\rx_u$. From the expression of $\rx_u=\tilde{\bf H}_u\tilde{\bf H}_u^H-{\bf h}_u{\bf h}_u$, there is only one negative eigenvalue in $\rx_u$, i.e., $r_R$. Substituting
$\bar\pu_u^1$ into the objective function, we have
\begin{align}
-J(\bar\pu_u^1)=-J(\bar\pu_u^0)+\sum\limits_{i=1}^{R-1}\epsilon_i^2 r_i>-J(\bar\pu_u^0).
\end{align}
Hence, the possible value of $\bar\pu_u$ cannot make the objective function keep dropping. Therefore,  when $f_u\bar{\bf v}_{u,R}$ is also on the surface of constraints, the closed-form solution of $\pu_u$ is solved as $f_u\bar{\bf v}_{u,R}$.
\section{Proof for Corollary \ref{C2}}\label{App3}
 According to Theorem \ref{T2}, the optimal point of ${\bf P}$ is on the surface of constraints. For simplicity, we assume that the optimal $\pr$ is achieved when $\psi'(\pr)=\rho$. In this case, the surface of $-J(\pr)=J_0$ and the surface of $\psi'(\pr)=\rho$ are tangent with each other. Otherwise, there will be a point of $\pr$ that satisfies $-J(\pr)=J_0$ and $\psi'(\pr)<\rho$, which is contradictory with Theorem 2. Therefore, the optimal point of $\pr$ is on the tangent plane of $-J(\pr)=J_0$ with $J_0$ being the minimal value.
\end{appendices}

\bibliographystyle{IEEEtran}

\end{document}